\documentclass[aps,reprint,pra,nofootinbib,superscriptaddress,amsmath,amssymb]{revtex4-2}
\usepackage{subfiles}

\usepackage[utf8]{inputenc}
\usepackage{amsmath, amssymb, amsfonts,amsthm,amsopn}
\usepackage{dsfont}
\usepackage{graphicx}
\usepackage{thm-restate}
\usepackage[normalem]{ulem}
\usepackage{array}
\usepackage{amssymb}
\usepackage{comment}
\usepackage{tikz}
\usepackage{mathtools}
\usepackage{xcolor}
\usepackage{braket}
\usepackage{bm}

\usepackage[breaklinks]{hyperref}
\hypersetup{
    colorlinks,
    linkcolor={red!40!black},
    citecolor={blue!50!black},
    urlcolor={blue!20!black}
}
\usepackage[capitalise]{cleveref} 
\crefformat{equation}{Eq.~(#2#1#3)}
\crefformat{section}{Sec.~#2#1#3} 
\Crefformat{equation}{Equation~(#2#1#3)}
\crefformat{figure}{Fig.~#2#1#3}
\crefrangeformat{equation}{Eqs.~#3(#1)#4--#5(#2)#6}
\Crefformat{section}{Section~#2#1#3}

\DeclareMathOperator{\tr}{tr}

\renewcommand{\ket}[1]{\mathinner{|#1\rangle}}

\newcommand{\dyad}[1]{|#1\rangle \!  \langle #1|}

\renewcommand{\a}{\alpha}

\newtheorem{theorem}{Theorem}
\newtheorem*{theorem*}{Theorem}
\newtheorem{proposition}[theorem]{Proposition}
\newtheorem{lemma}[theorem]{Lemma}
\newtheorem{definition}[theorem]{Definition}
\newtheorem{corollary}[theorem]{Corollary}

\newtheorem{example}[theorem]{Example}
\newtheorem{remark}[theorem]{Remark}

\newtheorem*{problem*}{Problem}
\newtheorem*{question*}{Question}

\newtheorem*{result*}{Result}

\providecommand{\ignore}[1]{}

\newif\ifcmnt
\cmnttrue
\ifdefined\cmntsoff\cmntfalse\fi

\ifcmnt

\providecommand{\aucmnt}[1]{#1}

\providecommand{\aucmnt}[1]{#1}

\else
    \providecommand{\aucmnt}[1]{}

\fi

\newcommand{\one}{\mathds{1}}

\ignore{

\theoremstyle{definition}

\numberwithin{equation}{section}

\renewcommand\bra[1]{{\langle{#1}|}}
\renewcommand\ket[1]{{|{#1}\rangle}}

\numberwithin{equation}{section}
\newtheorem{theorem}{Theorem}[section]
\newtheorem*{theorem*}{Theorem}
\newtheorem{proposition}[theorem]{Proposition}
\newtheorem*{proposition*}{Proposition}
\newtheorem{lemma}[theorem]{Lemma}
\newtheorem*{lemma*}{Lemma}

}

\begin{document}

\title{Constraints on recovering quantum information after erasure}

\author{Mohammad A. Alhejji}
\affiliation{Center for Quantum Information and Control, University of New Mexico, Albuquerque, NM 87131, USA}
\affiliation{Department of Physics and Astronomy, University of New Mexico, Albuquerque, NM 87131, USA}

\author{Noah Lordi}
\affiliation{Center for Quantum Information and Control, University of New Mexico, Albuquerque, NM 87131, USA}
\affiliation{Department of Physics and Astronomy, University of New Mexico, Albuquerque, NM 87131, USA}

\author{Omkar Baraskar}
\affiliation{Cheriton School of Computer Science, University of Waterloo, Waterloo, ON N2L 3G1, Canada}

\author{Ariel Shlosberg}
\affiliation{Center for Quantum Information and Control, University of New Mexico, Albuquerque, NM 87131, USA}
\affiliation{Department of Physics and Astronomy, University of New Mexico, Albuquerque, NM 87131, USA}

\author{Emanuel Knill}
\affiliation{Applied and Computational Mathematics Division, National Institute of Standards and Technology, Boulder, CO, 80305, USA}
\affiliation{Department of Physics, University of Colorado, Boulder, CO, 80309, USA}
\affiliation{Center for Theory of Quantum Matter, University of Colorado, Boulder, CO, 80305, USA}

\date{\today}

\begin{abstract}
Suppose that we encode quantum information into \(n\) physical carriers. With what probabilities can we perfectly recover the quantum information after sets of carriers have been erased? This question arises in the study of quantum erasure channels, for which M.B. Hastings conjectured that if each carrier is independently erased with probability \(t \geq 1/2\), then it is impossible to recover the quantum information with asymptotically vanishing error and probability \(> 1-t\). We prove this conjecture by establishing new no-cloning constraints on realizable recovery probabilities. 
\end{abstract}

\maketitle

\noindent\textit{Introduction}. A striking feature of quantum information is that it cannot be cloned~\cite{Wootters1982}, from which it follows that the minimum distance of a quantum code is less than half its length plus one.

In general, a quantum code of length \(n\) and dimension \(K > 1\) is
a physical operation that encodes a \(K\)-dimensional quantum system
into \(n\) information-carrier quantum systems. The most commonly
studied quantum codes are isometric, that is, subspace codes. The
minimum distance of a quantum code is the largest number \(d\) such
that there is a physical operation that perfectly recovers the quantum
information after erasure of any carrier set of size
\(\leq d-1\)~\cite{Rains1998}. Equivalently, \(d\) is the largest
number such that if it is possible to detect which carriers suffered
errors and there are at most \(d-1\) of them, then the quantum
information can be perfectly recovered with probability
\(1\). We refer to a quantum code whose length is \(n\), dimension is \(K\), and minimum distance is \(d\) as an \(((n,K,d))\) code. We consider carriers of finite but otherwise
unrestricted dimensions.

The no-cloning bound on minimum distance, \(d < n/2 + 1\), is realizable in the sense that for every \(n\) and \(d\) that satisfy the bound and every \(K\), there exists, at the cost of large carrier dimension, an \(((n, K, d))\) code~\cite{Grassl2004}. Minimum distance is a measure of worst-case performance; perfect recovery must be possible after erasure of any set of \(d-1\) or fewer carriers. Hence, characterizations of code performance based on minimum distance can be coarse. For a given \(((n,K,d))\) code, perfect recovery may be possible after erasure of some carrier sets of size \(> d-1\). This is the case for the \(((n, 4, \sqrt{n/2})) \) qubit toric code~\cite{KITAEV2003}. It can correct for erasure of any contractible arrangement of qubits and such arrangements can contain many more qubits than \(\sqrt{n/2}-1\).

For a fine-grained characterization of code performance, we can associate each carrier set \(S\) with the largest probability of perfectly recovering the quantum information from the carriers in \(S\). This invites the question: given a tuple of \(2^n\) probabilities \(q= (q_{S})_{S}\), is there a quantum code of length \(n\) such that for each carrier set \(S\) there exists a probabilistic physical operation on the carriers in \(S\) occurring with probability \(q_S\) that perfectly recovers the quantum information? We say that such a code realizes the tuple \(q\). For example, a quantum code of length \(n\) and minimum distance \(d\) realizes a tuple \(q\) for which \(q_{S}=1\) for all carrier sets \(S\) of size at least \(n-d+1\). In this Letter, we investigate the sets \(\mathcal{Q}_{n}\) of tuples \(q\) for which there exists a quantum code of length \(n\) that realizes \(q\).

A simple observation is that if every carrier is erased, then the probability of recovering the quantum information is zero. Hence, the empty set coordinate is zero for all tuples in \(\mathcal{Q}_n\). No-cloning implies that each tuple \(q \in \mathcal{Q}_n\) satisfies \(\sum_{S_i} q_{S_i} \leq 1\) whenever the sets \(S_i\) are disjoint; this is a corollary of \cref{eq: recovery operator for partition bound} below. We show that for \(n\leq 4\), \(\mathcal{Q}_{n}\) is exactly the set of tuples with nonnegative coordinates that satisfy these constraints. For \(n>4\), we show that there must be additional constraints. Furthermore, we prove new no-cloning constraints on recovery probabilities. Specifically, we prove that the sum of recovery probabilities associated with a family of sets is bounded from above by the Lovász number~\cite{Lovasz1979} of the family's disjointness graph.

Our investigation of set recovery probabilities is motivated by the erasure-simulation conjecture of M.B. Hastings~\cite{hastings2014notes}, which is a significant open problem in the study of quantum channel simulation~\cite{Bennett2014}. The quantum erasure channel \(\mathcal{E}_{t}\) with erasure probability \(t\) acts as the identity channel with probability \(1-t\) and otherwise outputs an erasure flag orthogonal to the input state space. We say that \(\mathcal{E}_{t}\) can simulate \(\mathcal{E}_{r}\) if there are sequences of quantum channels \(\mathcal{A}_{n}\) and \(\mathcal{D}_{n}\) such that \(\mathcal{D}_n \mathcal{E}^{\otimes n}_t \mathcal{A}_n\) converges to \(\mathcal{E}_{r}\) in diamond norm as \(n\) goes to infinity. Not surprisingly, if \(r\geq t\), then \(\mathcal{E}_{t}\) can simulate \(\mathcal{E}_r\). In the reverse case, it is a fundamental result of the theory of quantum error correction~\cite{Bennett1997} that for \(t<1/2\), \(\mathcal{E}_{t}\) can simulate the identity channel \(\mathcal{E}_{0}\). When \(t \geq 1/2\), does there exist \(r<t\) such that \(\mathcal{E}_{t}\) can simulate \(\mathcal{E}_{r}\)? Hastings posed this question in Ref.~\cite{hastings2014notes} with the goal of gaining insight into simulability as an equivalence relation on quantum channels. For \(r < 1/2\), the answer is ``no'' because a quantum channel with zero quantum capacity such as \(\mathcal{E}_t\) cannot simulate a quantum channel with positive quantum capacity such as \(\mathcal{E}_r\). In particular, \(\mathcal{E}_{1/2}\) cannot simulate an erasure channel with erasure probability less than \(1/2\).  By applying a no-cloning argument, Hastings established that for each positive integer \(m\), \(\mathcal{E}_{1-1/m}\) cannot simulate an erasure channel with erasure probability less than \(1-1/m\). According to the erasure-simulation conjecture, this result generalizes from \(\mathcal{E}_{1-1/m}\) to \(\mathcal{E}_{t}\) for all \(t \geq 1/2\). We prove this conjecture by applying our results constraining \(\mathcal{Q}_{n}\).

\vspace{\baselineskip} 
\noindent\textit{Basic properties of \(\mathcal{Q}_n\)}. We label carriers by elements of the set \([n] := \{1,2,\ldots,n\}\). Accordingly, we use subsets of \([n]\) to refer to carrier sets and to label the coordinates of tuples.
The set of realizable tuples \(\mathcal{Q}_{n}\) implicitly depends on \(K\), the dimension of the encoded quantum system, which we assume
to be strictly greater than \(1\). For \(K'\leq K\), a \(K'\)-dimensional system can be encoded as a subspace of a
\(K\)-dimensional system. Therefore, the set \(\mathcal{Q}_n\) does not grow with \(K\); we do not know if it
shrinks or stays the same. In any case, our results are independent of
\(K\) for \(K > 1\).

The set \(\mathcal{Q}_{n}\) is convex. For \(q^{(0)}, q^{(1)} \in \mathcal{Q}_n\) and \(w \in [0,1]\), the convex mixture \((1-w) q^{(0)} + w q^{(1)}\) is realized by a mixture of quantum codes realizing \(q^{(0)}\) and \(q^{(1)}\). In more detail, the encoder samples a binary random variable \(X\) such that \(\Pr (X=1) = w\) and encodes to realize \(q^{(X)}\). In addition to this probabilistic encoding, the encoder expands each carrier by attaching one qubit in the state \(\dyad{X}\). To recover the encoded system from a nonempty carrier set, the decoder applies the recovery operation appropriate for the code indicated by the attached qubits. Convexity follows because the carrier dimensions are unbounded.  For the same reason, it is possible that \(\mathcal{Q}_{n}\) is not topologically closed in general.

Next, we observe that the decoder can choose to apply a quantum erasure channel after a recovery
operation to decrease the probability of recovery.  This implies
that if \(q\) is in \(\mathcal{Q}_n\), then every tuple \(q'\)
satisfying \(0 \leq q'_S \leq q_S\) for all \(S \subseteq [n]\) is
also in \(\mathcal{Q}_n\). In other words, the set \(\mathcal{Q}_n\)
is downward closed in the coordinate-wise order on tuples with
nonnegative coordinates.

Every \(q \in \mathcal{Q}_n\) is realized by an isometric quantum channel encoding the quantum information into a subspace code. Recall that a quantum channel is a linear completely positive trace preserving map and that a quantum channel \(\mathcal{U}\) is isometric if and only if \(\mathcal{U} (\cdot) = U (\cdot) U^\dagger\) for some isometry \(U\). Stinespring's dilation theorem~\cite{Paulsen2009-jq} implies that every quantum channel \(\mathcal{N}\) can be dilated to an isometric quantum channel \(\mathcal{U}_{\mathcal{N}}\) in the sense that \(\mathcal{N} = \tr_E \circ \; \mathcal{U}_{\mathcal{N}}\), where \(E\) denotes an auxiliary system output by \(\mathcal{U}_{\mathcal{N}}\) commonly known as the environment of the dilation. Let \(\mathcal{A}\) be an encoding channel in some realization of \(q\) and let \(\mathcal{U}_{\mathcal{A}}\) be an isometric dilation of \(\mathcal{A}\). To realize \(q\), the encoder encodes according to \(\mathcal{U}_{\mathcal{A}}\) and attaches the environment system to one of the carriers in the output of \(\mathcal{A}\), increasing this carrier's dimension. The decoder discards the environment system if it has not been erased already and then applies recovery operations which together with \(\mathcal{A}\) realize \(q\). If \(U_\mathcal{A}\) is an isometry such that \(\mathcal{U}_\mathcal{A} (\cdot) = U_\mathcal{A} (\cdot) U_\mathcal{A}^\dagger\), then the subspace code is the image of \(U_\mathcal{A}\).

No-cloning implies that for \(q\in \mathcal{Q}_{n}\), if \(S_1,S_2\) are disjoint carrier sets, then \(q_{S_1}+q_{S_2}\leq 1\). This can be seen from the following thought experiment. After encoding the quantum information to realize \(q\), the recovery operations corresponding to \(S_1\) and \(S_2\) may be run in parallel. This is because \(S_1\) and \(S_2\) are disjoint. If \(q_{S_1} + q_{S_2} > 1\), then there would be a positive probability that the quantum information is perfectly recovered from the two disjoint sets simultaneously!  As this would violate no-cloning, we conclude that \(q_{S_1} + q_{S_2} \leq 1\). More generally, for every partition \(\mathcal{P}\) of \([n]\) and every tuple \(q \in \mathcal{Q}_n\), it must hold that
\begin{align}
\label{ee: partition ineq}
    \sum_{S \in \mathcal{P}} q_S \leq 1. 
\end{align}
This inequality is a formal corollary of  \cref{eq: recovery operator for partition bound} below. We denote by \(\mathcal{T}_{n}\) the set of tuples \(x\) which satisfy \(x_\varnothing = 0\), \(0 \leq x_S \leq 1\) for all \(S \subseteq [n]\), and for every partition \(\mathcal{P}\) of \([n]\), \(\sum_{S \in \mathcal{P}} x_S \leq 1.\) The inclusion \(\mathcal{Q}_n \subseteq \mathcal{T}_n\) is immediate from the foregoing considerations.

\vspace*{\baselineskip}
\noindent\textit{\(\mathcal{Q}_{n}=\mathcal{T}_{n}\) for \(n\leq 4\)}. We use Sage~\cite{sagemath} to compute the extreme  points of \(\mathcal{T}_n\), then we show how to realize every one of them, and appeal to the convexity of \(\mathcal{Q}_n\). We consider the \(n=1,2,3,4\) cases in order. The extreme points of \(\mathcal{T}_{1}\) are \(x^{(0)}\) with \(x^{(0)}_{\varnothing}=0\) and \(x^{(0)}_{\{1\}}=0\), and \(x^{(1)}\) with \(x^{(1)}_{\varnothing}=0\) and \(x^{(1)}_{\{1\}}=1\).  The first is trivially realizable, and the second is realized by isometrically encoding the quantum information into a \(K\)-dimensional carrier.

The number of the extreme points of \(\mathcal{T}_{n}\) grows very rapidly with \(n\) so we prune them as follows. First, we may consider extreme points up to permutations of the carriers. Second, if \(x^{(0)}\) and \(x^{(1)}\) are extreme points of \(\mathcal{T}_n\) such that \(x^{(0)}_S \leq x_S^{(1)}\) for all \(S \subseteq [n]\) and \(x^{(1)}\) is realizable, then, because \(\mathcal{Q}_n\) is downward closed in the coordinate-wise order, \(x^{(0)}\) is realizable as well. So we need only address the realizability of \(x^{(1)}\). Third, it is enough to only consider extreme points \(x \in \mathcal{T}_n\) with the property that for every \( i \in [n]\) there exists some \(S \subseteq [n]\) such that \(x_{S \cup \{ i \}} > x_S\). Suppose that \(x\) is an extreme point that does not have this property. Without loss of generality, assume that it is \(n\) which satisfies \(x_{S \cup \{n\}} \leq x_{S}\) for all \(S \subseteq [n]\). Then, the tuple \(x':=(x_{S})_{S\subseteq[n-1]}\) is in \(\mathcal{T}_{n-1}\). Assuming that \(\mathcal{Q}_{n-1} = \mathcal{T}_{n-1}\) holds, this implies that \(x'\) is in \(\mathcal{Q}_{n-1}\). To realize \(x\), we may use a quantum code realizing \(x'\) to encode the quantum information into the first \(n-1\) carriers and append the \(n\)th carrier in some fixed state. For the recovery operation from \(S \subseteq [n]\), we discard the last carrier, apply a recovery operation to realize \(x'_{S \setminus \{n\}}\), and follow with an erasure channel if necessary to reduce the probability of recovery.

We find that \(\mathcal{T}_2\) has \(6\) extreme points, none of which remain after pruning. \(\mathcal{T}_3\) has \(40\) extreme points and only one, \(x^{(3)}\), remains after pruning. It is given by \(x_S^{(3)} = 1\) if \(|S| \geq 2\), and \(x_S^{(3)} = 0\) otherwise. This extreme point is realized by a \(((3, K, 2))\) code. Lastly, \(\mathcal{T}_4\) has \(1376\) extreme points and only one, \(x^{(4)}\), remains after pruning. It is given by \(x_S^{(4)} = 1\) if \(|S| \geq 3\) or \(|S| = 2\) and \(1 \in S\), and \(x_S^{(4)} = 0\) otherwise. We can realize this extreme point by encoding the quantum information into an adjusted \(((5,K,3))\) code where the first two of the five qudits are bundled together to form the first carrier.

\vspace*{\baselineskip}
\noindent\textit{Partition constraints on \(\mathcal{Q}_{n}\)}. The probabilistic physical recovery operations in the definition of \(\mathcal{Q}_{n}\) are modeled with quantum operations from the state spaces of the carrier-set systems to the decoded \(K\)-dimensional system. A quantum operation is a linear, completely positive map that does not increase the trace of any positive semidefinite operator in its domain. For every quantum operation \(\mathcal{R}\), there exists a positive semidefinite operator \(R \leq \one\) and a quantum channel \(\mathcal{J}\) such that \(\mathcal{R}(\cdot)=\mathcal{J}(R (\cdot) R)\). Since quantum channels are trace preserving, for a state \(\rho\) we can identify \(\tr(R\rho R)\) with the probability that \(\mathcal{R}\) is successfully applied to \(\rho\). Quantum operations can be realized with post-selection. See App.~\ref{sec: lemma} for more details.

For every realizable tuple of recovery probabilities \(q \in \mathcal{Q}_n\), there exists a subspace code \(\mathcal{C}\) with the property that for every carrier set \(S\), there exists a quantum operation \(\mathcal{R}_{S}\) acting on the carriers in \(S\) that perfectly recovers the quantum information encoded in \(\mathcal{C}\) with probability \(q_{S}\). According to the above, there exists a positive semidefinite operator \(R_{S} \leq \one\) and a quantum channel \(\mathcal{J}_{S}\) on the carriers in \(S\) such that \(\mathcal{R}_{S}(\cdot)= \mathcal{J}_{S}(R_{S} (\cdot) R_{S})\). This implies that \(\tr(R_{S}\rho R_{S}) = q_S\) for every state \(\rho\) supported in \(\mathcal{C}\). In particular, for every unit vector \(\ket{\psi} \in \mathcal{C}\), \(q_{S}=\bra{\psi}R_{S}^{2}\ket{\psi} = || R_S \ket{\psi}||^2 \). As a consequence, when \(q_S > 0\), the operator \(R_{S}/\sqrt{q_{S}}\) restricted to \(\mathcal{C}\) is an isometry. Furthermore, the image \(R_S \mathcal{C}\) is a \(K\)-dimensional subspace that can correct for erasure of the carriers in the complement \(S^c\), see \cref{eq: cond quantum error correction} in the appendix. The quantum error correction conditions for \(\tr_{S^c} \) imply that for any operator \(X_{S^c}\) on the carriers in \(S^c\), the restriction of \(X_{S^c}\) to \(R_S \mathcal{C}\) is proportional to an isometry. Since \(R_S\) is proportional to an isometry on \(\mathcal{C}\), this implies that the product \(R_S X_{S^c}\) is proportional to an isometry on \(\mathcal{C}\).

Let \(S_1\) and \(S_2\) be disjoint carrier sets. Then, we can apply the quantum operations \(\mathcal{R}_{S_1}\) and
\(\mathcal{R}_{S_2}\) in parallel to states supported in \(\mathcal{C}\). For a unit vector \(\ket{\psi} \in \mathcal{C}\), the probability that both operations are applied to \(\dyad{\psi}\) is \(\tr(R_{S_1}R_{S_2}\dyad{\psi} R_{S_2}R_{S_1})\) which equals \(|| R_{S_1} R_{S_2} \ket{\psi}||^2\). By the above, \(R_{S_1} R_{S_2}\) is proportional to an isometry on \(\mathcal{C}\) and so the probability \(|| R_{S_1} R_{S_2} \ket{\psi}||^2\) is independent of \(\ket{\psi}\). Suppose toward a contradiction that it is positive. After applying \(\mathcal{R}_{S_1}\) to \(\dyad{\psi}\), the state on the
\(K\)-dimensional output of \(\mathcal{R}_{S_1}\) is pure and factors from the state of the carriers in \(S_1^{c}\). In particular,
except for scale, the state of the output of \(\mathcal{R}_{S_1}\) is unchanged if we then apply \(\mathcal{R}_{S_2}\) to the carriers in \(S_2\), and vice versa. Consequently, after applying the quantum operations in parallel, the marginal state on each of the two \(K\)-dimensional outputs is pure and isometrically equivalent to \(\dyad{\psi}\).  No-cloning prevents this from being the case for all \(\ket{\psi} \in \mathcal{C}\), from which we conclude that \(R_{S_1}R_{S_2}\mathcal{C}=0\). It follows from this that for every partition \(\mathcal{P}\) of \([n]\), it holds that
\begin{align}
  \sum_{S \in \mathcal{P}} q_S &= \bra{\psi} \sum_{S \in \mathcal{P}} R_S^2 \ket{\psi} \leq \bra{\psi} \sum_{S \in \mathcal{P}} R_S \ket{\psi} \notag\\
&\leq \braket{\psi|\psi}^{1/2} \bra{\psi}\sum_{S,S'\in\mathcal{P}}R_{S}R_{S'}\ket{\psi}^{1/2} \notag \\
& =  \bra{\psi}\sum_{S \in \mathcal{P}} R_S^2 \ket{\psi}^{1/2} = \left(\sum_{S \in \mathcal{P}} q_S \right)^{1/2}. \label{eq: recovery operator for partition bound}
\end{align}
This chain of inequalities implies that \(\sum_{S \in \mathcal{P}} q_S \leq 1\), with equality if and only if for all \(\ket{\psi} \in \mathcal{C}\), \(\ket{\psi} = \sum_{S \in \mathcal{P}} R_S \ket{\psi}\) and  \(\bra{\psi} R_S -R_S^2 \ket{\psi}= 0\) for all \(S \in \mathcal{P}\).

\vspace*{\baselineskip}
\noindent\textit{For \(n>4\), \(\mathcal{Q}_{n}\subsetneq\mathcal{T}_{n}\)}. Now we exhibit a tuple in \(\mathcal{T}_5\) that is not in \(\mathcal{Q}_5\). We refer to a set containing exactly \(k\) elements as a \(k\)-set. Let \(u \in \mathcal{T}_5\) denote the tuple satisfying \(u_{S} = 1/2\) when \(S \subseteq \{1,2,3,4,5\}\) is a \(2\)-set and \(u_S = 0\) otherwise.  Suppose toward a contradiction that this tuple is realizable. Then there would exist a subspace code \(\mathcal{C}\) and recovery operators \(R_S\), one for each \(2\)-set \(S \subseteq \{1,2,3,4,5\}\), that each occur with probability \(1/2\) on \(\mathcal{C}\). For a unit vector \(\ket{\psi} \in \mathcal{C}\),  by the previous paragraph, for each pair of disjoint \(2\)-sets \(S\) and \(S'\), it holds that
\begin{align}
\ket{\psi}  = R_S \ket{\psi} + R_{S'} \ket{\psi}.
\label{eq: proj decomp}
\end{align}
We obtain from Eq.~\eqref{eq: proj decomp} the equalities \(\ket{\psi} - R_{\{4,5\}} \ket{\psi} = R_{\{1,2\}} \ket{\psi} = R_{\{2,3\}} \ket{\psi}\) and \(\ket{\psi}-R_{\{5,1\}} \ket{\psi} = R_{\{2,3\}} \ket{\psi} = R_{\{3,4\}} \ket{\psi}\), from which we conclude the absurd 
\begin{align}
R_{\{1,2\}} \ket{\psi} = R_{\{3,4\}} \ket{\psi}!
\end{align}
Thus, \(u \notin \mathcal{Q}_5\). To show that \(\mathcal{Q}_{n}\subsetneq\mathcal{T}_{n}\) for \(n > 5\), it suffices to append zeros to \(u\) to obtain \(u'\in \mathcal{T}_{n}\) and apply the same argument.

\vspace*{\baselineskip}
\noindent \textit{New constraints on \(\mathcal{Q}_{n}\)}. To preclude such absurd tuples, we prove new constraints on recovery probabilities. Consider a family \(\mathcal{F}\) of subsets of \([n]\). Its disjointness graph, \(G_{\mathcal{F}}\), is the graph whose vertices are the elements of the family and whose edges are between disjoint elements. As an example, \cref{fig: K52} depicts the disjointness graph of the family of \(2\)-sets in \(\{1,2,3,4,5\}\). The complement of a disjointness graph has edges between intersecting elements and is called an intersection graph. An orthonormal representation of a graph is a map that assigns each vertex to a unit vector in a Hilbert space such that for every two non-adjacent vertices, the corresponding unit vectors are orthogonal~\cite{Lovasz1979}.  Given \(q \in \mathcal{Q}_n\), we can obtain from no-cloning an orthonormal representation of the intersection graph of \(\mathcal{F}\). Fix a unit vector \(\ket{\psi}\) in a code \(\mathcal{C}\) realizing \(q\). For each \(S \in \mathcal{F}\) such that \(q_S > 0\), define the unit vector \(\ket{v_S} : = R_S \ket{\psi} / || R_S \ket{\psi} ||\). Because \(0 \leq R_S \leq \mathds{1}\), we have 
\begin{align}
\label{eq: estimate for R and v}
\braket{\psi | R_{S}^2 | \psi} &= \frac{\braket{\psi | R_{S}^2 | \psi}^2}{\braket{\psi | R_{S}^2 | \psi}} \notag \\
&\leq \frac{\braket{\psi | R_{S} | \psi}^2}{\braket{\psi | R_{S}^2 | \psi}} = \braket{\psi|v_S} \!\!  \braket{v_S | \psi}.
\end{align}
For \(S \in \mathcal{F}\) with \(q_S = 0\), we assign mutually orthogonal unit vectors \(\ket{v_S}\)  that are orthogonal to the Hilbert space of the carriers. Since \(\braket{v_S |v_{S'}} = 0\) for each disjoint \(S\) and \(S'\) in \(\mathcal{F}\), these unit vectors form an orthonormal representation of the intersection graph of \({\mathcal{F}}\).  We use this construction to estimate the sum of recovery probabilities corresponding to elements in \(\mathcal{F}\): 
\begin{align}
\sum_{S \in \mathcal{F}} q_S &= \sum_{S \in \mathcal{F}} \braket{\psi | R_{S}^2 | \psi} \notag\\
&\leq \sum_{S \in \mathcal{F}} \braket{\psi|v_S} \!\!  \braket{v_S | \psi} \notag\\
&\leq\max_{\ket{\psi'}, \ket{v'_S}} \sum_{S \in \mathcal{F}} |\braket{\psi' | v'_S}|^2, \label{eq: formula for Lovasz num}
\end{align}
where the maximization is over all orthonormal representations \(\{ \ket{v_S'} \}_{S \in \mathcal{F}}\) of the intersection graph of \(\mathcal{F}\) and all unit vectors \(\ket{\psi'}\). The result of this maximization, \(\vartheta(G_{\mathcal{F}})\), is the Lovász number of the disjointness graph \(G_{\mathcal{F}}\)~\cite[Thm.~5]{Lovasz1979}. Originally, the Lovász number was shown to equal the maximum value of the objective in \cref{eq: formula for Lovasz num} over all real orthonormal representations. A proof that maximizing the objective over all complex orthonormal representations yields the same number is in Ref.~\cite{Marrero2015}. 
Thus, it holds that
\begin{align}
\label{eq: estimate for recovery prob}
  \sum_{S \in \mathcal{F}} q_S \leq \vartheta (G_{\mathcal{F}}).
\end{align}
If \(\mathcal{F}\) is a family of disjoint sets, then \(\vartheta(G_{\mathcal{F}}) =1\). Thus, the constraint \cref{ee: partition ineq} is a special case of \cref{eq: estimate for recovery prob}.

\begin{figure}[t!]
\centering
    \begin{tikzpicture}[
    vertex/.style={circle, draw=black, thick, fill=blue!10, minimum size=7mm, inner sep=0pt, font=\small},
    edge/.style={thick, black}, scale = 0.7
]

\foreach \i/\a/\lblout/\lblin in {
    0/90/{1,2}/{3,5},
    1/162/{3,4}/{2,5},
    2/234/{1,5}/{2,4},
    3/306/{2,3}/{1,4},
    4/18/{4,5}/{1,3}%
} {
    \node[vertex] (v\i) at (\a:3) {$\lblout$};
    \node[vertex] (u\i) at (\a:1.5) {$\lblin$};
}

\draw[edge] (v0) -- (v1) -- (v2) -- (v3) -- (v4) -- (v0);

\draw[edge] (u0) -- (u2) -- (u4) -- (u1) -- (u3) -- (u0);

\foreach \i in {0,1,2,3,4} {
    \draw[edge] (v\i) -- (u\i);
}
\end{tikzpicture}
\caption{The disjointness graph of the family of \(2\)-sets in \(\{1,2,3,4,5\}\) is known as the Petersen graph. Its Lovász number is \(4\) which equals the size of its largest independent set.}
\label{fig: K52}
\end{figure}
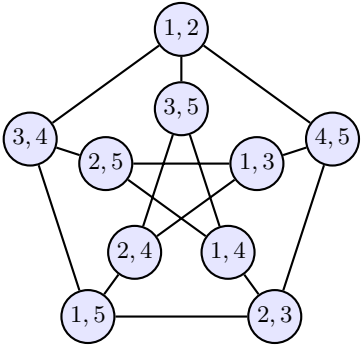

Let \(K_{n,k}\) denote the disjointness graph of the family of \(k\)-sets in \([n]\). For \(k \leq n/2\), the Lovász number of \(K_{n,k}\) is known~\cite[Thm.~13]{Lovasz1979}:
\begin{align}
\label{eq: kneser theta}
\vartheta(K_{n,k}) = \binom{n-1}{k-1}. 
\end{align}
Below and in the Supplemental Material, we use this along with \cref{eq: estimate for recovery prob} to prove the erasure-simulation conjecture.

\vspace*{\baselineskip}
\noindent\textit{Proof of the erasure-simulation conjecture}. Fix a real number \(t\) that satisfies \(1/2 < t < 1\). We wish to show that if \(\mathcal{E}_{t}\) can simulate \(\mathcal{E}_{r}\), then \(r\geq t\). The \(t=1/2\) and \(t=1\) cases are addressed in Ref.~\cite{hastings2014notes}. Since increasing the carrier dimensions can only make simulation easier, we allow the carrier dimensions to be greater than \(K\).

We first argue that we can restrict to isometric encodings. For this, consider a simulation  \(\mathcal{D}_{n}\mathcal{E}_{t}^{\otimes n}\mathcal{A}_{n}\) of \(\mathcal{E}_{r}\). If \(\mathcal{A}_{n}\) is not isometric, then we can dilate \(\mathcal{A}_{n}\) to an isometric channel and embed the finite-dimensional environment of the dilation into a finite number, \(\ell\), of auxiliary erasable carrier systems. Denote this isometric channel by \(\mathcal{U}_{n+\ell}\). We compose the decoding channel \(\mathcal{D}_n\) with a partial trace over the auxiliary carrier systems to arrive at a decoding channel \(\mathcal{D}'_{n+\ell}\). Then, for each \(n\) it holds that  \(\mathcal{D}_{n}\mathcal{E}_{t}^{\otimes n}\mathcal{A}_{n} = \mathcal{D}'_{n+\ell} \mathcal{E}_t^{\otimes (n+\ell)} \mathcal{U}_{n+\ell}.\)

We address here the special case of exact simulation, which captures the heart of the matter. A complete proof of the erasure-simulation conjecture is in the Supplemental Material. Let us suppose that there exist quantum channels \(\mathcal{U}_n\) and \(\mathcal{D}_n\), where the \(\mathcal{U}_n\)'s are isometric, such that 
\begin{align}
\label{eq: zero error equality}
\mathcal{E}_{r_n} = \mathcal{D}_n \mathcal{E}_t^{\otimes n} \mathcal{U}_n \quad \text{and} \quad r_n \rightarrow r \in [0,1], 
\end{align}
as \(n\rightarrow\infty\). Let \(\dyad{e}\) be the erasure flag projector and define \(P := \mathds{1} - \dyad{e}\). Expanding the tensor product on the right-hand side yields
\begin{align}
\label{eq: zero erro euqality exp}
\mathcal{E}_{r_n}  = \sum_{S \subseteq [n]} (1-t)^{|S|} t^{n-|S|} \mathcal{D}_{n}\; \dyad{e}^{\otimes S^c} \otimes \tr_{S^c} \mathcal{U}_n.
\end{align}
By composing both sides of \cref{eq: zero erro euqality exp} with the action \(X \mapsto P X P \), we obtain that \(1-r_{n}\) times the identity channel equals a convex combination of quantum operations. As the identity channel is an extreme point of the set of completely positive maps, each of these quantum operations is proportional to the identity channel. For \(S \subseteq [n]\), the proportionality constant, \(q_S\),  is the recovery probability conditional on erasure of the carriers in \(S^c\). Immediately, we have \(1 - r_n = \sum_{S \subseteq [n]} (1-t)^{|S|} t^{n - |S|} q_S.\)

Let \(\delta > 0\) denote a number small enough so that \(1-t + \delta < 1/2\). We break the sum \(\sum_{S \subseteq [n]} (1-t)^{|S|} t^{n - |S|} q_S\) into a typical part accounting for the sets \(S \subseteq [n]\) such that \(n (1-t) - n\delta \leq |S| \leq n(1-t) + n \delta\) and an atypical part accounting for the rest that we denote by \(\omega (n)\). Now we estimate using \cref{eq: estimate for recovery prob} and \cref{eq: kneser theta}:
\begin{align}
\label{eq: estimates for zero error case}
    1-r_n &= \sum_{\text{typical } k} (1-t)^{k} t^{n-k}\sum_{S \subseteq [n], |S| = k} q_S + \omega (n) \notag\\
    &\leq \sum_{\text{typical } k} (1-t)^{k} t^{n-k} \binom{n-1}{k-1} + \omega (n) \notag\\
    &\leq 1-t + \omega(n).
\end{align}
By the law of large numbers, \(\omega(n) \rightarrow 0\) and so \(1-r \leq 1-t\). This completes the proof of the special case.

In the result and proof above, we required \(n\rightarrow\infty\). We point out that the result implies in particular that if for some \(n\), some encoding channel  \(\mathcal{A}'_n\) and some decoding channel \(\mathcal{D}'_n\), it holds that 
\begin{align}
\label{eq: channel equality end}
\mathcal{E}_r = \mathcal{D}'_n \mathcal{E}_t^{\otimes n} \mathcal{A}'_n, 
\end{align}
then \(1-r \leq 1-t\).  This is because we can append arbitrarily many auxiliary carriers initialized in some fixed state, encode into the first \(n\) carriers using \(\mathcal{A}'_n\) and decode by using \(\mathcal{D}'_n\) on the first \(n\) carriers and a partial trace on the auxiliary carriers.

\vspace*{\baselineskip}
\noindent\textit{Concluding remarks}. For five carriers or more, a complete understanding of realizable recovery probabilities is lacking. We highlight a curious inequality that holds for points in \(\mathcal{Q}_5\). The disjointness graph of the family of sets \(\{\{1,2\}, \{2,3\}, \{3,4\}, \{4,5\}, \{5,1\}\}\) is a pentagon and the Lovász number of a pentagon is \(\sqrt{5}\). So, for every point \(q \in \mathcal{Q}_5\), we have
\begin{align}
q_{\{1,2\}} + q_{\{2,3\}} + q_{\{3,4\}} + q_{\{4,5\}} + q_{\{5,1\}}
\leq \sqrt{5}. 
\end{align}
Is there a quantum code that saturates this inequality? Another question we leave for future work is whether there exist constraints on realizable recovery probabilities beyond no-cloning.

\vspace*{\baselineskip}
\noindent\textit{Data availability}. The Sage script we used to compute the extreme points of \(\mathcal{T}_n\) for \(n \leq 4\) is available at~\cite{repo2026}.
\begin{acknowledgments}
\label{ack}
While this paper was in the final editing stages, we learned that Rabsan Galib Ahmed and Graeme Smith independently proved the erasure-simulation conjecture in Ref.~\cite{ahmed2026distributed}. The authors thank Cole Maurer for helpful discussions. MA, NL, and AS acknowledge support from U.S. National Science Foundation Grant PHY-2116246. This work includes contributions of the National Institute of Standards and Technology (NIST), which are not subject to  U.S. copyright. The U.S. Government is authorized to reproduce and  distribute reprints for governmental purposes notwithstanding any copyright annotation thereon. The use of trade names and software does not imply endorsement by the US government, nor does it imply these are necessarily the best available for the purpose used here. This document has not been peer reviewed but has been cleared by NIST for release. The authors acknowledge use of Anthropic's Claude (Fable 5), a generative artificial intelligence tool, to research and interpret the literature of Lovász numbers in order to bound from above the optimal value of the program in Eq.~\ref{eq:e sdp for theta}. Proof strategies produced through the use of Claude were refined by the authors and then used to prove Prop.~\ref{prop: program upper bound} in the Supplemental Material. 
\end{acknowledgments}
\vspace*{\baselineskip}
\vspace*{\baselineskip}

\bibliography{references}

\appendix

\section{Quantum operation lemma}
\label{sec: lemma}

\begin{lemma}
\label{lem: aux qubit for instrument}
Let \(\mathcal{R}\) denote a quantum operation on a finite-dimensional quantum system. There exists a positive semidefinite operator \(R\) and a quantum channel \(\mathcal{J}\) such that \(\mathcal{R} (\cdot ) = \mathcal{J} ( R (\cdot) R)\). 
\end{lemma}
\begin{proof}
Let  \(\{R_1, \ldots, R_m\}\) be a Kraus representation for \(\mathcal{R}\). Define \(R\) to be the positive semidefinite operator \((\sum_{i=1}^m R_i^{\dagger} R_i)^{1/2}\). For each \(i \in [m]\), let \(J_i\) denote the product \(R_i R^{-1}\). If \(R\) is not invertible, then \(R^{-1}\) is the Moore-Penrose inverse of \(R\). Then
\begin{align}
\label{eq: q inst into measurement and channel}
\sum_{i=1}^m J_i R (\cdot) R J_i^{\dagger} = \sum_{i=1}^m R_i (\cdot) R_i^\dagger = \mathcal{R}(\cdot). 
\end{align}
Furthermore, letting \(P_{R >0}\) denote the projection onto the support of \(R\), we have
\begin{align}
\label{eq: trace presrving cond for J}
\sum_{i=1}^m J_i^{\dagger} J_i = \sum_{i=1}^m R^{-1} R_i^{\dagger} R_i R^{-1} = P_{R > 0}. 
\end{align}
If \(R\) is not invertible, we can include an extra Kraus operator \(J_{m+1} : = ( \mathds{1} - P_{R > 0})^{1/2}\) so that \(\mathcal{J} (\cdot) := \sum_{i=1}^{m+1} J_i (\cdot) J_i^\dagger\) is a quantum channel. This extra operator annihilates the image of \(R\) so it does not affect \cref{eq: q inst into measurement and channel}.  \end{proof}

This lemma gives an operational recipe to implement \(\mathcal{R}\). We evolve the system and an auxiliary qubit initialized in \(\ket{0}_{\text{aux}}\) by a unitary that maps \(\ket{\alpha} \otimes \ket{0}_{\text{aux}}\)  to 
\begin{align}
\label{eq: uni with aux}
R \ket{\alpha} \otimes \ket{0}_{\text{aux}} + (\mathds{1} - R^2)^{1/2} \ket{\alpha} \otimes \ket{1}_{\text{aux}},
\end{align}
where \(\ket{\alpha}\) denotes an arbitrary unit vector in the Hilbert space of the system. Then we measure the auxiliary qubit to find out whether \(R\) occurred or not. If so, then we implement the quantum channel \(\mathcal{J}\).

If \(\mathcal{R}\) is supposed to perfectly recover quantum information after some noise channel \(\mathcal{N}\) with Kraus representation \(\{N_j\}_{j=1}^\ell\), then the positive semidefinite operator \(R\) appears in the quantum error correction conditions. Let \(\mathcal{C}\) denote a subspace code such that when restricted to this code, \(\mathcal{R} \mathcal{N}\) equals \(p\) times the identity channel for some \(0<p\leq 1\). Then for every unit vector \(\ket{\alpha} \in \mathcal{C}\), the quantum error correction conditions 
\begin{align}
\label{eq: cond quantum error correction}
\bra{\alpha} N_{i}^\dagger R^2 N_{j} \ket{\alpha} &= T_{i,j} \; \forall i, j \in [\ell]
\end{align}
must hold. The probability of recovery \(p\) equals \(\sum_{i=1}^\ell T_{ii}\). The problem of finding quantum recovery operations with the highest possible recovery probability reduces to a semidefinite programming problem~\cite{Kuku23, erasure23}.

\clearpage
\newpage
\onecolumngrid

\documentclass[Main.tex]{subfiles}
\ifSubfilesClassLoaded{
  \allowdisplaybreaks
}{}

\begin{document}

\ifSubfilesClassLoaded{
  \title{Supplemental Material: Proof of the erasure-simulation conjecture}
  \author{Mohammad A. Alhejji}
  \affiliation{Center for Quantum Information and Control, University of New Mexico, Albuquerque, NM 87131, USA}
  \affiliation{Department of Physics and Astronomy, University of New Mexico, Albuquerque, NM 87131, USA}

  \author{Noah Lordi}
  \affiliation{Center for Quantum Information and Control, University of New Mexico, Albuquerque, NM 87131, USA}
  \affiliation{Department of Physics and Astronomy, University of New Mexico, Albuquerque, NM 87131, USA}

  \author{Omkar Baraskar}
  \affiliation{Cheriton School of Computer Science, University of Waterloo, Waterloo, ON N2L 3G1, Canada}

  \author{Ariel Shlosberg}
  \affiliation{Center for Quantum Information and Control, University of New Mexico, Albuquerque, NM 87131, USA}
  \affiliation{Department of Physics and Astronomy, University of New Mexico, Albuquerque, NM 87131, USA}

  \author{Emanuel Knill}
  \affiliation{Applied and Computational Mathematics Division, National Institute of Standards and Technology, Boulder, CO, 80305, USA}
  \affiliation{Department of Physics, University of Colorado, Boulder, CO, 80309, USA}
  \affiliation{Center for Theory of Quantum Matter, University of Colorado, Boulder, CO, 80305, USA}
  \date{\today}

  \maketitle
  
}{
  \stepcounter{section}%
  \setcounter{equation}{0}%
  \section*{Supplemental Material: Proof of the erasure-simulation conjecture}
  \label{app:em}
  \let\section\subsection
}

\ifSubfilesClassLoaded{%
  \onecolumngrid
  \appendix
}{}

For \(s \in [0,1]\), the erasure channel \(\mathcal{E}_{s}\) maps operators on a
finite-dimensional Hilbert space to operators on the extension of this space
by an orthogonal erasure flag state \(\ket{e}\) and is defined by
\begin{align}
\mathcal{E}_{s}(X) := (1-s) X +s \tr(X)\dyad{e}.
\end{align}
Let \(t\) and \(r\) denote two real numbers that satisfy \(1/2 < t < 1\) and \(0 \leq r < 1\).  Suppose that there exist isometric encoding channels \(\mathcal{U}_n\) and decoding channels \(\mathcal{D}_n\) such that \(\mathcal{C}_{n}:=\mathcal{D}_n \mathcal{E}_t^{\otimes n} \mathcal{U}_n\) converges to \(\mathcal{E}_r\) in diamond norm as \(n\rightarrow\infty\). The input space of \(\mathcal{E}_r\) is \(K\)-dimensional for \(K > 1\). The input space of \(\mathcal{E}_t\) is arbitrary but assumed to be finite-dimensional for simplicity. Our goal is to prove that \(1-t \geq 1-r\). The assumption that the encoding channels are isometric is without loss of generality, as we noted in the main text.

We start by expanding the tensor product in \(\mathcal{D}_n \mathcal{E}_t^{\otimes n} \mathcal{U}_n\) to obtain
\begin{align}
\label{eq:e conv comb of channels 1}
\mathcal{D}_n \mathcal{E}_t^{\otimes n} \mathcal{U}_n = \sum_{S \subseteq [n]} t^{n-|S|} (1-t)^{|S|} \mathcal{D}_{n}\;\dyad{e}^{\otimes S^c} \otimes \tr_{S^c} \mathcal{U}_n.
\end{align}
We may, if necessary, follow \(\mathcal{D}_n \left(\mathcal{E}_t^{\otimes n} \mathcal{U}_n\right)\) with a pinching channel that kills any coherence between the input space and the erasure flag space without increasing its diamond distance from \(\mathcal{E}_r\). For each carrier set \(S \subseteq [n]\), let \(\mathcal{D}_{S}\) denote the conditional decoding channel on the carriers in \(S\) defined by the action \(X_S \mapsto \mathcal{D}_n ( \dyad{e}^{\otimes S^c} \otimes X_S)\). We rewrite the convex combination of channels in Eq.~\eqref{eq:e conv comb of channels 1}:
\begin{align}
\label{eq:e conv comb of channels 2}
\mathcal{D}_n \mathcal{E}_t^{\otimes n} \mathcal{U}_n = \sum_{S \subseteq [n]} t^{n-|S|} (1-t)^{|S|}  \mathcal{D}_{S} \tr_{S^c} \mathcal{U}_n.
\end{align}
Because the erasure flag space and the input space are decohered, we may take each conditional
decoding channel \(\mathcal{D}_{S}\) to decompose orthogonally according to erasure
as \(\mathcal{M}_S \oplus \mathcal{M}_{S}^{e}\), where \(\mathcal{M}_S\) and \(\mathcal{M}_S^{e}\) are the quantum operations
defined by
\begin{align}
  \mathcal{M}_{S}^{e}(X_{S})
  :=\bra{e}\mathcal{D}_{S}(X_{S})\ket{e}\dyad{e} \quad \text{and} \quad 
  \mathcal{M}_{S}(X_{S})
  := \mathcal{D}_{S}(X_{S})-\mathcal{M}_{S}^{e}(X_{S}).
\end{align}
These are quantum operations on the carriers in \(S\). We denote the
composite channel \(\mathcal{D}_{S} \tr_{S^c} \mathcal{U}_n\) by
\(\mathcal{N}_S\). When applying maps defined on a subsystem to a state of a larger system, we implicitly tensor them with an identity map to extend the operations to the larger system.

We adjoin a \(K\)-dimensional reference space to the input space of
\(\mathcal{E}_r\) and prepare a maximally entangled state
\(\Phi = \dyad{\Phi}\) between the input space and this reference. We
introduce \(\Phi_{n}= \dyad{\Phi_n}\) where
\(\Phi_{n}=\mathcal{U}_n (\Phi)\), which is the pure encoded version
of \(\Phi\).  Here, \(\mathcal{U}_{n}\) is applied to the input-space
factor.  Let \(P\) denote the projector onto the orthogonal complement
of the erasure flag \(\ket{e}\) in the output space of
\(\mathcal{E}_r\).  For a carrier set \(S\), define
\(q_S := \tr (P \mathcal{N}_S(\Phi)) = \tr (\mathcal{M}_S(\Phi_n))\).
The overall probability of not erasing when given input \(\Phi\) is given by
\begin{align}
\label{eq:e conv for probabilities}
1-r_n &:= \sum_{S \subseteq [n]} t^{n-|S|} (1-t)^{|S|}  q_S.
\end{align}
We refer to \(1-r_{n}\) as a recovery probability.  Diamond-norm
convergence of \(\mathcal{C}_{n}\) to \(\mathcal{E}_{r}\) implies that
\(1-r_{n}\) converges to \(1-r\). 

For a carrier set \(S\) such that \(q_{S}>0\), define the
(entanglement) recovery fidelity as
\begin{align}
\label{eq:e recovery fidelity}
f_{S} := \frac{1}{q_S} \tr( \Phi \mathcal{N}_S(\Phi)) = \frac{1}{q_S} \tr( \Phi \mathcal{M}_S(\Phi_n)).
\end{align}
Define the (entanglement) recovery error as \(\varepsilon_{S}:=(1-f_{S})\).  For \(q_{S}=0\), define \(f_{S}:=0\) and
\(\varepsilon_{S}:=1\). See Fig.~\ref{fig:1} for a circuit depiction of the erasure simulation. 

\begin{figure}[t!]
    \centering
    \includegraphics[width=0.80\textwidth]{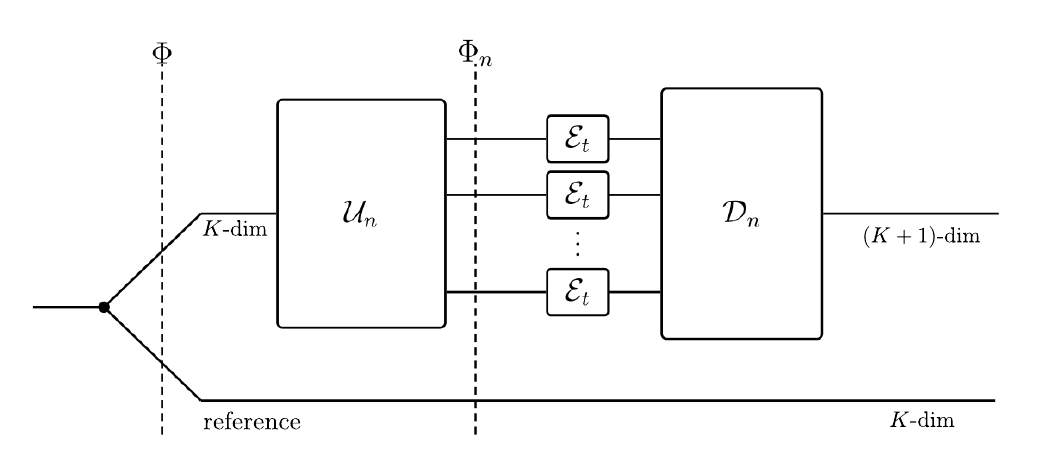}
    \caption{A circuit depiction of the quantum channel \(\mathcal{D}_n \mathcal{E}_t^{\otimes n} \mathcal{U}_n\) applied to one half of the maximally entangled state \(\Phi\).} 
    \label{fig:1}
\end{figure}

We can separate the performance of the erasure simulation by the size of the
subsets \(S\). Let \(\binom{[n]}{k}\) denote the family of \(k\)-sets in \([n]\). Define the conditional-on-\(k\) recovery probability \(q_k\) as the sum \(\sum_{S \in \binom{[n]}{k}} q_S\). For \(k\) such that \(q_k \neq 0\), define the recovery fidelity and recovery error
according to
\begin{align}
\quad f_{k} :=\frac{1}{q_{k}}\sum_{S\in\binom{[n]}{k}}q_{S}f_{S} \quad \text{and} \quad   \varepsilon_{k} :=\frac{1}{q_{k}}\sum_{S\in\binom{[n]}{k}}q_{S}\varepsilon_{S}.
\end{align}
For \(q_{k}=0\), we set \(f_{k} :=0\) and \(\varepsilon_{k} := 1\). We use \(o_n(1)\) as shorthand for a quantity which tends to \(0\) when \(n\) tends to \(+\infty\), all else held constant. Diamond-norm convergence implies that \(\bra{\Phi}\mathcal{D}_{n}\mathcal{E}_{t}^{\otimes n}\mathcal{U}_{n}(\Phi)\ket{\Phi} = (1-r)+o_{n}(1)\). Therefore,
\begin{align}
  \sum_{k=0}^n t^{n-k}(1-t)^{k}q_{k}f_{k}
  &= \sum_{S \subseteq [n]}t^{n-|S|}(1-t)^{|S|}q_{S}f_{S}\notag\\
  &=\bra{\Phi}\mathcal{D}_{n}\mathcal{E}_{t}^{\otimes
    n}\mathcal{U}_{n}(\Phi)\ket{\Phi} = (1-r)+o_{n}(1),
\end{align}
and
\begin{align}
  \sum_{k=0}^n t^{n-k}(1-t)^{k}q_{k}\varepsilon_{k}
  &=  \sum_{k=0}^n t^{n-k}(1-t)^{k}q_{k}(1-f_{k})\notag\\
  &= \sum_{k=0}^n t^{n-k}(1-t)^{k}q_{k} -   \sum_{k=0}^n t^{n-k}(1-t)^{k}q_{k}f_{k}
    \notag\\
  &= (1-r_{n})-(1-r) +o_{n}(1) = o_{n}(1).
    \label{eq:e qeps limit}
\end{align}

By the lemma in the appendix of the main text, for each \(S \subseteq [n]\), there
exists a positive semidefinite operator \( R_S\) satisfying \(R_S \leq \one\) and a
quantum channel \(\mathcal{J}_{S}\) acting on the carriers in \(S\) such that
for all operators \(X_{S}\) on the carriers in \(S\),
\begin{align}
\label{eq:e recovery operator}
  \mathcal{M}_S (X_S) = P ( \mathcal{D}_{S} (X_S)) P = \mathcal{J}_{S} ( R_S (X_S) R_S).
\end{align}
We can express \(q_S\) as an expectation of \(R_S^2\):
\begin{align}
\label{eq:e recovery probability for S}
q_S = \tr(\mathcal{M}_S (\Phi_n)) &= \tr( P \mathcal{D}_{S} (\Phi_n) P) \notag \\
&= \tr( \mathcal{J}_{S} (R_S \, \Phi_n \, R_S) ) \notag \\
&= \bra{\Phi_n} R_S^2 \ket{\Phi_n}\notag \\
&=  ||R_{S} \ket{\Phi_n}||^2.
\end{align}

Let us consider what happens when the decoding channels corresponding
to two disjoint sets \(S_1, S_2 \subseteq [n]\) are applied in
parallel.  The parallel application results in the quantum channel
\begin{align}
\label{eq:e B12 def}
\mathcal{B}_{1 2} := \tr_{(S_1 \cup S_2)^c} \mathcal{D}_{S_1} \mathcal{D}_{S_2} \mathcal{U}_n.
\end{align}
We label the output systems of \(\mathcal{D}_{S_{b}}\) by \(b\) for
\(b=1,2\).  Then
\begin{align}
\label{eq:e marginal channels}
\tr_2 \mathcal{B}_{12} = \mathcal{D}_{S_1} \tr_{S_1^c} \mathcal{U}_n = \mathcal{N}_{S_1} \quad \text{and} \quad \tr_1 \mathcal{B}_{12}  = \mathcal{D}_{S_2} \tr_{S_2^c} \mathcal{U}_n = \mathcal{N}_{S_2}.
\end{align}
We label the reference system by \(3\) and, for readability, place it between the two output systems of \(\mathcal{B}_{12}\). The probability of simultaneous recovery is
\begin{align}
\label{eq:e p12 def}
p_{12} := \tr( P_1 \otimes \one_3 \otimes P_2 \; \mathcal{B}_{1 2} (\Phi)) = \tr(P_1 \otimes \one_3 \otimes P_2 \; \mathcal{D}_{S_1} \mathcal{D}_{S_2} (\Phi_n)),
\end{align}
where for \(b=1,2\), \(P_{b}\) is the projector onto the space
orthogonal to the erasure flag space of the output of
\(\mathcal{D}_{S_{b}}\).  Let the operator \(E_{3}^{(1)}\)
denote the marginal of \(\mathcal{M}_{S_1}^{e} (\Phi_n)\) on the reference system
\(3\), and similarly for \(E_{3}^{(2)}\).  These operators satisfy
\(\tr (E_3^{(b)}) = 1-q_{S_b}\).  Consider the tripartite state
\(\Gamma_{132} := \mathcal{B}_{12} (\Phi)\). This state has 
marginals
\begin{align}
\Gamma_{13} = q_{S_1} \Delta_{13} + \dyad{e}_1 \otimes E_3^{(1)} \quad \text{and} \quad 
\Gamma_{32} = q_{S_2} \Delta_{32} + E_3^{(2)} \otimes \dyad{e}_2,
\end{align}
where \(\Delta_{13}\) denotes \(\frac{1}{q_{S_1}}\tr_{S_{1}^{c}}\mathcal{M}_{S_1} (\Phi_n)\) and \(\Delta_{32}\) denotes \(\frac{1}{q_{S_2}} \tr_{S_{2}^{c}}\mathcal{M}_{S_2} (\Phi_n)\). See Fig.~\ref{fig:2} for a circuit depiction of how \(\Gamma\) arises. 

For the next step, we apply the following monogamy of entanglement lemma, which
may be considered to be a quantified form of no-cloning.

\begin{lemma}
\label{lem:e monogamy of entanglement lemma}
Let \(\rho_{ABC}\) be a tripartite quantum state on three tensor copies of a \(d\)-dimensional space. For every pair of maximally entangled states \(\Psi_{AB}\) and \(\Psi_{BC}\), it holds that
\begin{align}
\label{eq:e monogamy of entanglement}
\tr ( \rho_{ABC} \; \Psi_{AB} \otimes \one_C) + \tr(\rho_{ABC} \; \one_A \otimes \Psi_{BC}) \leq 1 + \frac{1}{d}.
\end{align}
\end{lemma}
\begin{proof}
The left-hand side of Eq.~\eqref{eq:e monogamy of entanglement} is bounded from above by the largest eigenvalue of the projector sum \(\Psi_{AB} \otimes \one_C + \one_A \otimes \Psi_{BC}\), which equals \(1\) plus the operator norm of the product \((\Psi_{AB} \otimes \one_C) (\one_A \otimes \Psi_{BC})\)~\cite{Vidav1977-hb}. An explicit calculation shows that the product equals
\(\frac{1}{d} \sum_{k=1}^d (\ket{\Psi}_{AB} \otimes \ket{k}_C) (\bra{k}_A \otimes \bra{\Psi}_{BC})\). Observe that there is a unitary
  \(V\) which is a composition of the swap of the \(A\) and \(C\) factors and local unitaries such that \(\frac{1}{d} \sum_{k=1}^d (\ket{\Psi}_{AB} \otimes \ket{k}_C) (\bra{k}_A \otimes \bra{\Psi}_{BC}) V = \frac{1}{d} \Psi_{AB} \otimes \one_C\). Since the operator norm is unitarily invariant, the claim follows from the fact that the operator norm of
  \(\frac{1}{d} \Psi_{AB} \otimes \one_C\) is \(\frac{1}{d}\).
\end{proof}

\begin{figure}[t!]
    \centering
    \includegraphics[width=0.85\textwidth]{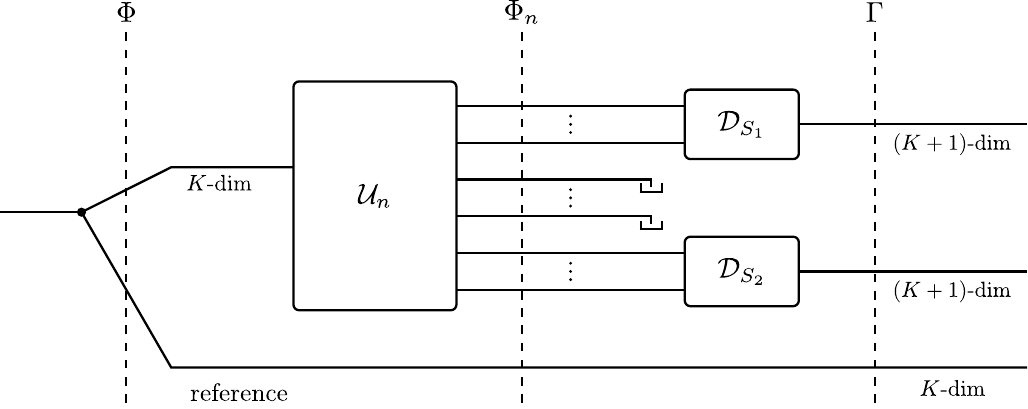}
    \caption{A circuit depiction of the quantum channel \(\mathcal{B}\) applied to one half of the maximally entangled state \(\Phi\) giving rise to the tripartite state \(\Gamma\).} 
    \label{fig:2}
\end{figure}

We apply the monogamy of entanglement inequality Eq.~\eqref{eq:e
  monogamy of entanglement} to upper bound the probability of
simultaneous recovery \(p_{12}\). Assume that \(p_{12}\) is nonzero,
as otherwise there is nothing to show. We can express \(p_{12}\) as
\begin{align}
  p_{12}= \tr ( ( \Phi_{13} + P_{1} \otimes \one_3 - \Phi_{13}) \otimes P_2\; \Gamma_{132}) ,
\end{align}
where \(\Phi_{13}\) is the reference-entangled state on systems
\(1,3\).  We apply the definitions of \(q_{S}\) and \(\varepsilon_{S}\)
to upper bound the right-hand side:
\begin{align}
  \tr ( ( \Phi_{13} + P_{1} \otimes \one_3 - \Phi_{13}) \otimes P_2 \;\Gamma_{132})
  &= \tr( \Phi_{13} \otimes P_2 \; \Gamma_{132}) + \tr( (P_{1} \otimes \one_3 - \Phi_{13}) \otimes P_2  \;\Gamma_{132}) \notag \\
  &\leq  \tr( \Phi_{13} \otimes P_2 \;\Gamma_{132}) + \tr( (P_{1} \otimes \one_3 - \Phi_{13}) \otimes \one_2 \;\Gamma_{132}) \notag \\
  &=\tr( \Phi_{13} \otimes P_2 \;\Gamma_{132})
    + \tr(\mathcal{M}_{S_1}(\Phi_n)) - \tr(\Phi \mathcal{M}_{S_{1}}(\Phi_{n}))\notag\\
  &= \tr( \Phi_{13} \otimes P_2 \;\Gamma_{132}) + q_{S_1}\varepsilon_{S_{1}}. \label{eq: p12 bound}
\end{align}
To obtain the second line, we used \(P_{1} \otimes \one_3 - \Phi_{13}\geq 0\)
and \(P_{2}\leq \one_{2}\).
By symmetry, we also have
\begin{align}
\label{eq: symmetry p12 bound}
  p_{12}\leq \tr( P_1 \otimes \Phi_{32} \; \Gamma_{132}) + q_{S_{2}}\varepsilon_{S_{2}}.
\end{align}
Define the normalized state
\[
  \tilde{\Gamma}_{132} := \frac{1}{p_{12}}(P_1 \otimes \one_3 \otimes P_2) \, \Gamma_{132} \, (P_1 \otimes \one_3 \otimes P_2).
\]
Then \(\tr( \Phi_{13} \otimes P_2 \; \Gamma_{132})=p_{12} \tr( \Phi_{13} \otimes P_2 \; \tilde{\Gamma}_{132})\) and \(\tr( P_1 \otimes \Phi_{32} \; \Gamma_{132})=p_{12} \tr( P_1 \otimes \Phi_{32} \; \tilde{\Gamma}_{132})\). By averaging the two
bounds \eqref{eq: p12 bound} and \eqref{eq: symmetry p12 bound}, we obtain
\begin{align}
  p_{12} &\leq \frac{p_{12}}{2} \left( \tr( \Phi_{13} \otimes P_2\; \tilde{\Gamma}_{132}) + \tr( P_1 \otimes \Phi_{32} \;\tilde{\Gamma}_{132}) \right) + (q_{S_{1}}\varepsilon_{S_{1}}+q_{S_{2}}\varepsilon_{S_{2}})/2.
\end{align}
By Lem.~\ref{lem:e monogamy of entanglement lemma}, this implies
\begin{align}
  p_{12} \leq \frac{p_{12}}{2} (1 + \frac{1}{K}) +
  (q_{S_{1}}\varepsilon_{S_{1}}+q_{S_{2}}\varepsilon_{S_{2}})/2.
\end{align}
Rearranging terms yields
\begin{align}
  p_{12} \leq \frac{q_{S_{1}}\varepsilon_{S_{1}}+q_{S_{2}}\varepsilon_{S_{2}}}{1 - {1}/{K}}.
\end{align}

In terms of the positive semidefinite operators \(R_{S_1}\) and \(R_{S_2}\) associated with disjoint \(S_1, S_2\), we have thus shown that
\begin{align}
  \label{eq: overlap bound for disjoint}
  p_{12} = \tr(P_1 \otimes \one_3 \otimes P_2 \; \mathcal{D}_{S_1} \mathcal{D}_{S_2} (\Phi_n)) = \big\| R_{S_1} R_{S_2} \ket{\Phi_n} \big\|^2 = \bra{\Phi_n} R_{S_1}^2 R_{S_2}^2 \ket{\Phi_n} \leq
  \frac{q_{S_{1}}\varepsilon_{S_{1}}+q_{S_{2}}\varepsilon_{S_{2}}}{1 - {1}/{K}}.
\end{align}

The bounds on \(p_{12}\) for disjoint \(S_{1}, S_{2}\) imply
constraints on the overall recovery probability. For this we need to
establish a relationship between recovery probabilities
\(q_{S}\) and bounds on the simultaneous recovery probabilities. To
obtain such a relationship we first extend the operators \(R_{S}^{2}\)
to projectors on a larger Hilbert space.  Let \(\Pi\) be the projector
onto the Hilbert space of the \(n\) carriers.  For each \(S\), use
\(0\leq R_{S}^{2}\leq\Pi\) to construct a projector \(Q_{S}\) on a
direct-sum extension of the carrier space by a Hilbert space
\(\mathcal{H}_{S}\) such that \(R_{S}^{2}=\Pi Q_{S}\Pi\), where the
Hilbert spaces \(\mathcal{H}_{S}\) are mutually orthogonal. The
construction ensures that
\(q_{S}=\bra{\Phi_{n}}R_{S}^{2}\ket{\Phi_{n}} =
\bra{\Phi_{n}}Q_{S}\ket{\Phi_{n}}\) and for distinct \(S_{1}, S_{2}\),
\(\bra{\Phi_{n}}R_{S_{1}}^{2} R_{S_{2}}^{2}\ket{\Phi_{n}} =
\bra{\Phi_{n}}Q_{S_{1}}Q_{S_{2}}\ket{\Phi_{n}}\).  For disjoint
\(S_{1}, S_{2}\) we get
\begin{align}
  \bra{\Phi_{n}}Q_{S_{1}}Q_{S_{2}}\ket{\Phi_{n}}
  &=\bra{\Phi_{n}}R_{S_{1}}^{2} R_{S_{2}}^{2}\ket{\Phi_{n}}\notag\\
  &\leq \frac{q_{S_{1}}\varepsilon_{S_{1}}+q_{S_{2}}\varepsilon_{S_{2}}}{1 - {1}/{K}}.
\end{align}
Of course, we also have the bound \(\bra{\Phi_{n}}Q_{S_{1}}Q_{S_{2}}\ket{\Phi_{n}}\leq 1\). 

We can apply the following property of Gram matrices to the Gram
matrix of the vectors \(Q_{S}\ket{\Phi_{n}}\).
\begin{lemma}
\label{lem:e block psd}
Let \(Q_1, \ldots, Q_p\) be projectors and let \(\ket{v}\) be a unit vector in their common domain. The Gram matrix \(M := ( \bra{v} Q_i Q_j \ket{v})_{i, j=1}^p\) and its diagonal \(m := ( \bra{v} Q_i \ket{v})_{i=1}^p\) have the property that
\begin{align}
\label{eq:e block psd ineq}
\begin{pmatrix}
1 & m^\dagger \\
m & M
\end{pmatrix} \geq 0.
\end{align}
Conversely, for every matrix \(\tilde{M}\) with diagonal \(\tilde{m}\) such that 
\begin{align}
\label{eq:e block matrix psd converse}
\begin{pmatrix}
1 & \tilde{m}^\dagger \\
\tilde{m} & \tilde{M}
\end{pmatrix} \geq 0,
\end{align}
there exist a unit vector \(\ket{\tilde{v}}\) and projectors \(\tilde{Q}_1,\ldots, \tilde{Q}_{p}\) such that \(\tilde{M} = (\bra{\tilde{v}} \tilde{Q}_i \tilde{Q}_j \ket{\tilde{v}})_{i, j=1}^p\). 
\end{lemma}
\begin{proof}
The block matrix in Eq.~\eqref{eq:e block psd ineq} is the Gram matrix of the vectors \(\ket{v}, Q_1 \ket{v}, \ldots, Q_p \ket{v}\), and so it is positive semidefinite. For the converse, Eq.~\eqref{eq:e block matrix psd converse} implies that there exist vectors \(\ket{w_0},\ldots,\ket{w_p}\) such that \(\braket{w_0|w_0} = 1\), \(\braket{w_j|w_0} = \braket{w_j|w_j}\) and \(\braket{w_i|w_j} = \tilde{M}_{ij}\) for all \(i, j \in [p]\). For \(j \in [p]\) such that \(\ket{w_j} = 0\), define \(\tilde{Q}_j := 0\). Otherwise, define \(\tilde{Q}_j := \frac{\dyad{w_j}}{\braket{w_j | w_j}}\), \(\ket{\tilde{v}} := \ket{w_0}\), and observe that \(\tilde{Q}_j \ket{\tilde{v}} = \frac{1}{\braket{w_j|w_j}} \braket{w_j | w_0} \ket{w_j} = \ket{w_j}\). 
\end{proof}

Consider the Gram matrices
\(N_{n,k} := ( \bra{\Phi_n} Q_{I} Q_{J} \ket{\Phi_n})_{I, J \in
  \binom{[n]}{k}}\) for \(k\in\{0,\ldots,n\}\). We have
\(\tr (N_{n,k})=q_{k}\) and
\(\tr (N_{n,k})\leq \binom{n}{k}\).  Lem.~\ref{lem:e
  block psd} and the bounds for disjoint subsets of the carriers imply
that \(N_{n,k}\) is a feasible point of the following
program:
\begin{align}
\label{eq:e sdp for theta}
\text{Maximize:} \quad & \tr(X) \notag \\
\text{Subject to:} \quad & X_{II} = x_I \quad\text{for all} \quad I \in \binom{[n]}{k}, \notag \\
& |X_{IJ}| \leq  \frac{q_{I}\varepsilon_{I}+q_{J}\varepsilon_{J}}{1 - {1}/{K}} \quad\text{for all}  \quad I, J \in \binom{[n]}{k} \;\; \text{s.t.} \;\; I \cap J = \varnothing, \\
& \begin{pmatrix}
1 & x^\dagger \\
x & X
\end{pmatrix} \geq 0. \notag
\end{align}
It follows then that the contribution of \(\tr (N_{n,k})\) to the overall recovery
probability is upper bounded by the program's value.

For the case of exact simulation, \(\varepsilon_{S}=0\) for all \(S\) and for
\(I, J\) disjoint, the upper bound on \(|X_{IJ}|\) is \(0\).  In
this case, the program's value is the Lov\'asz number of the
disjointness graph of \(\binom{[n]}{k}\)~\cite{grotschel2012geometric}.
See SDP2 in Ref.~\cite{galli2017lovasz}.  For
\(k>n/2\), there are no disjoint sets in \(\binom{[n]}{k}\) and the Lov\'asz number and
optimal value of the program is just \(\binom{n}{k}\). For \(k\leq n/2\), the Lov\'asz number is \(\binom{n-1}{k-1}\), as shown
by Lov\'asz in Ref.~\cite[Thm. 13]{Lovasz1979}.

Define
\(\overline\varepsilon_{k}:=
\binom{n}{k}^{-1}\sum_{S\in\binom{[n]}{k}}q_{S}\varepsilon_{S}=\binom{n}{k}^{-1}q_{k}\varepsilon_{k}\) to be the
average of \(q_{I}\varepsilon_{I}\) over \(k\)-sets \(I\).

\begin{proposition}
\label{prop: program upper bound} 
Let \(X\) be a feasible point of the program~\eqref{eq:e sdp for theta} with \(k\) satisfying \(1\leq k\leq n/2\). Then
  \begin{align}
  \label{eq: epsilon lovasz bound}
    \tr(X) &\leq \frac{1}{2}\binom{n-1}{k-1}\left(
             1+\sqrt{1+\frac{n(n-k)}{k^{2}}8\overline\varepsilon_{k}/(1-1/K)}
             \right).
  \end{align}
\end{proposition}

\begin{proof}
  Let \(A_{n,k}\) be the adjacency matrix of the disjointness graph of
  \(\binom{[n]}{k}\). Then \(A_{n,k}\) is the real symmetric matrix whose \(I,J\)
  entry is \(1\) if \(I\cap J=\varnothing\) and \(0\)
  otherwise. We have
  \begin{align}
    \tr(A_{n,k}X)
    &= \sum_{I,J: I\cap J=\varnothing}X_{IJ}\notag\\
    &\leq \sum_{I,J:I\cap J=\varnothing}\frac{q_{I}\varepsilon_{I}+q_{J}\varepsilon_{J}}{1 - {1}/{K}}\notag\\
    &=2\binom{n}{k}\binom{n-k}{k}\frac{\overline\varepsilon_{k}}{1-1/K}.
  \end{align}
  According to the proof of Thm. 13 in Ref.~\cite{Lovasz1979},
  \(A_{n,k}\) has \(k+1\) distinct eigenvalues, which are given by
  \((-1)^{j} \binom{n-k-j}{k-j}\) for \(j \in \{0,1,\ldots,k\}\). The
  most negative eigenvalue of \(A_{n,k}\) is \(-\binom{n-k-1}{k-1}\),
  corresponding to \(j=1\).  Therefore we have
  \(B_{n,k}:=A_{n,k}+\binom{n-k-1}{k-1}\one \geq 0\).  By
  vertex-transitivity of the disjointness graph of \(\binom{[n]}{k}\), the all-ones vector
  \(\bm{1}\) is an eigenvector of \(A_{n,k}\) with
  eigenvalue \(\binom{n-k}{k}\), corresponding to \(j=0\). It follows that \(A_{n,k}\)
  and \(B_{n,k}\) commute with \(\bm{1}\bm{1}^{\dagger}\). Since the
  eigenvalue of \(B_{n,k}\) at \(\bm{1}\) is
  \(\binom{n-k}{k}+\binom{n-k-1}{k-1}=\binom{n-k-1}{k-1}\frac{n}{k}\),
  we have
  \(B_{n,k}\geq
  \binom{n-k-1}{k-1}\frac{n}{k}\binom{n}{k}^{-1}\bm{1} \bm{1}^{\dagger}
  =\binom{n-k-1}{k-1}\binom{n-1}{k-1}^{-1}\bm{1}\bm{1}^{\dagger}\).
  
  The feasibility of \(X\) implies that the Schur complement
  \(X - x x^\dagger\) is positive semidefinite, which implies that
    \(\tr(\mathbf{1} \mathbf{1}^\dagger\, X) \geq \tr (\mathbf{1} \mathbf{1}^\dagger x x^\dagger ) = \tr(X)^2\).
  Since in particular \(X\geq 0\), we obtain the inequality
  \(\tr(B_{n,k}\,X)\geq \binom{n-k-1}{k-1}\binom{n-1}{k-1}^{-1}
  \tr(\bm{1} \bm{1}^\dagger\, X)
  \geq \binom{n-k-1}{k-1}\binom{n-1}{k-1}^{-1}\tr(X)^{2}\).
  After expanding the definition of \(B_{n,k}\) and substituting identities, we get
  \begin{align}
    0
    &\geq \binom{n-k-1}{k-1}\binom{n-1}{k-1}^{-1}\tr(X)^{2} - \tr(B_{n,k}\,X)\notag\\
    &= \binom{n-k-1}{k-1}\binom{n-1}{k-1}^{-1}\tr(X)^{2} - \binom{n-k-1}{k-1}\tr(X)
      - \tr(A_{n,k}\,X)\notag\\
    &\geq \binom{n-k-1}{k-1}\binom{n-1}{k-1}^{-1}\tr(X)^{2} - \binom{n-k-1}{k-1}\tr(X)
      - 2\binom{n}{k}\binom{n-k}{k}\frac{\overline\varepsilon_{k}}{1-1/K}\notag\\
    &=\binom{n-k-1}{k-1}\binom{n-1}{k-1}^{-1}\left(
      \tr(X)^{2}-\binom{n-1}{k-1}\tr(X)- 2\binom{n}{k}\binom{n-1}{k-1}\frac{n-k}{k}\frac{\overline\varepsilon_{k}}{1-1/K}
      \right).
  \end{align}
  The expression in parenthesis is a quadratic polynomial in \(\tr(X)\) which has real roots, one of which is positive and the other is negative. Since the polynomial is not positive at \(\tr(X) = 0\), the upper bound on \(\tr(X)\) is obtained by finding its positive root:
  \begin{align}
    \tr(X)\leq \frac{1}{2}\binom{n-1}{k-1}\left(1+\sqrt{1+8\frac{(n-k)n}{k^{2}}\frac{\overline\varepsilon_{k}}{1-1/K}}\right).  \tag*{\qedhere}
  \end{align}
\end{proof}

We can now upper bound the overall recovery probability. Choose a
constant \(\delta>0\) such that \(1-t+\delta<1/2\) and
\(1-t-\delta>0\). Let \(D\) be the interval \(D:=[(1-t-\delta)n,(1-t+\delta)n]\). The bound is obtained with the following sequence of inequalities. The derivation of the sequence is explained below.
\begin{align}
  (1-r_{n})
  & = \sum_{k=0}^n t^{n-k}(1-t)^{k}q_{k} = \sum_{k=0}^n t^{n-k}(1-t)^{k}\tr (N_{n,k})\\
  &\leq \sum_{k\in D}
    t^{n-k}(1-t)^{k}\frac{1}{2}\binom{n-1}{k-1}\left(1
    +\sqrt{1+8\frac{(n-k)n}{k^{2}}\frac{\overline\varepsilon_{k}}{1-1/K}}\right)
    \label{eq:e final 1}\\
  &\hphantom{\leq\sum_{k\in[(1-t-\delta)n]}}
    + \sum_{k\notin D}t^{n-k}(1-t)^{k}\binom{n}{k}
    \notag\\
  & = (1-t)\sum_{k+1\in D}
    t^{n-k-1}(1-t)^{k}\binom{n-1}{k}\frac{1}{2}\left(1
    +\sqrt{1+8\frac{(n-k-1)n}{(k+1)^{2}}\frac{\overline\varepsilon_{k+1}}{1-1/K}}\right) + o_{n}(1)
    \label{eq:e final 2}\\
  &\leq
    (1-t)\sum_{k+1\in D}
    t^{n-k-1}(1-t)^{k}\binom{n-1}{k}\frac{1}{2}\left(1
    +\sqrt{1+8\frac{t+\delta}{(1-t-\delta)^{2}}\frac{\overline\varepsilon_{k+1}}{1-1/K}}\right) + o_{n}(1)
    \label{eq:e final 3}\\
    &\leq (1-t)\sum_{k=0}^{n-1}
    t^{n-k-1}(1-t)^{k}\binom{n-1}{k}\frac{1}{2}\left(1
+\sqrt{1+8\frac{t+\delta}{(1-t-\delta)^{2}}\frac{\overline\varepsilon_{k+1}}{1-1/K}}\right) + o_{n}(1)
      \label{eq:e final 4}\\
    &\leq (1-t)\frac{1}{2}\left(1
      +\sqrt{1+8\frac{t+\delta}{(1-t-\delta)^{2}}\sum_{k=0}^{n-1}
      t^{n-k-1}(1-t)^{k}\binom{n-1}{k}\frac{\overline\varepsilon_{k+1}}{1-1/K}}\right) + o_{n}(1)
      \label{eq:e final 5}\\
  &=(1-t)\frac{1}{2}\left(1
      +\sqrt{1+8\frac{t+\delta}{(1-t-\delta)^{2}}\sum_{k=0}^{n-1}
    t^{n-k-1}(1-t)^{k}\frac{k+1}{n}\frac{q_{k+1}\varepsilon_{k+1}}{1-1/K}}\right) + o_{n}(1)
    \label{eq:e final 6}\\
  &\leq(1-t)\frac{1}{2}\left(1
      +\sqrt{1+8\frac{t+\delta}{(1-t-\delta)^{2}(1-t)}\sum_{k=0}^{n}
    t^{n-k}(1-t)^{k}\frac{q_{k}\varepsilon_{k}}{1-1/K}}\right) + o_{n}(1)
    \label{eq:e final 7}\\
  &=(1-t)\frac{1}{2}\left(1
      +\sqrt{1+o_{n}(1)}\right) + o_{n}(1)
    =(1-t) + o_{n}(1).
\end{align}
To obtain line~\eqref{eq:e final 1}, we split the sum according to whether
\(k\in D\) or not. For \(k\notin D\), we applied the trivial bound
\(\tr (N_{n,k}) \leq \binom{n}{k}\). For \(k\in D\) we applied the bound
in Eq.~\eqref{eq: epsilon lovasz bound}. The next line~\eqref{eq:e final 2} is obtained by
applying the law of large numbers, according to which
\(\sum_{k\notin D}t^{n-k}(1-t)^{k}\binom{n}{k} = o_{n}(1)\).  At the
same time, we shifted the index \(k\) by \(1\), which causes no
difficulties for large enough \(n\).  For line~\eqref{eq:e final 3},
we replaced \((n-k-1)n/(k+1)^{2}\) by its upper bound according to the
constraint \(k+1\in D\), for which it suffices to substitute the lower
end of \(D\) for \(k+1\). Next, for line~\eqref{eq:e final 4}, we
added positive terms by extending the sum over all
\(k\in\{0,\ldots,n-1\}\).  To obtain line~\eqref{eq:e final 5}, we
interpret \(t^{n-k-1}(1-t)^{k}\binom{n-1}{k}\) as probabilities to
recognize the sum as an average of the expression after these
probabilities. This expression is a concave function of
\(\overline\varepsilon_{k+1}\). The average of the expression is therefore
upper bounded by the expression evaluated at the average.  At
line~\eqref{eq:e final 6} we replaced \(\binom{n-1}{k}\) by
\(\frac{k+1}{n}\binom{n}{k+1}\) and applied the defining identity
for \(\overline\varepsilon_{k+1}\). For the next line, we replaced \(\frac{k+1}{n}\) by \(1\), shifted \(k\) again, pulled out a factor of \(1/(1-t)\), and added a nonnegative \(k=0\) term to the sum under the
square root. For the last step we applied Eq.~\eqref{eq:e qeps  limit}. 

It follows that \(1-r_{n} = (1-r)+o_{n}(1) \leq (1-t)+o_{n}(1)\) so that
\(1-t\geq 1-r\), as promised.

\ifSubfilesClassLoaded{%
  \bibliographystyle{plainurl}
  \bibliography{references}
}{}
\end{document}

\end{document}